\documentclass[runningheads]{llncs}

\usepackage{todonotes}
\usepackage{hyperref}
\usepackage{amsmath, amssymb, amstext}
\usepackage{subcaption}
\usepackage{tikz}
\usetikzlibrary{automata}
\usepackage{cleveref}
\usepackage{ltl}
\usepackage{algorithm}
\usepackage[noend]{algpseudocode}
\usepackage{multicol}

\usepackage{makecell}
\newtheorem{probdef}{Problem}
\usepackage{wrapfig}
\usepackage{stmaryrd}

\author{Tom Baumeister\inst{1} \and Bernd Finkbeiner\inst{2} \and Hazem Torfah\inst{3}}

\institute{Saarland University, Saarland Informatics Campus\\
\email{s8tobaum@stud.uni-saarland.de}
\and
CISPA Helmholtz Center for Information Security\\
\email{finkbeiner@cispa.saarland}
\and
University of California, Berkeley, USA\\
\email{torfah@berkeley.edu}}

\begin{document}

\title{Explainable Reactive Synthesis\thanks{This work was partially supported by the Collaborative Research Center ``Foundations of Perspicuous Software Systems'' (TRR: 248, 389792660), the European Research Council (ERC) Grant OSARES (No. 683300), the DARPA Assured Autonomy program, the iCyPhy center, and by Berkeley Deep drive.}}
	
\maketitle

\begin{abstract}
Reactive synthesis transforms a specification of a reactive system, given in a temporal logic, into an implementation. The main advantage of synthesis is that it is automatic. The main disadvantage is that the implementation is usually very difficult to understand. In this paper, we present a new synthesis process that explains the synthesized implementation to the user. The process starts with a simple version of the specification and a corresponding simple implementation. Then, desired properties are added one by one, and the corresponding transformations, repairing the implementation, are explained in terms of counterexample traces. We present SAT-based algorithms for the synthesis of repairs and explanations. The algorithms are evaluated on a range of examples including benchmarks taken from the SYNTCOMP competition.


\end{abstract}

\section{Introduction}\label{sec:Introduction}
In reactive synthesis, an implementation of a reactive system is automatically constructed from its formal specification. Synthesis allows developers to define the behavior of a system in terms of a list of its desired high-level properties, delegating the detailed implementation decisions to an automatic procedure. However, the great advantage of synthesis, that it is \emph{automatic}, can also be an obstacle, because it makes it difficult for the user to \emph{understand} why the system reacts in a particular way. This is particularly troublesome in case the user has written an incorrect specification or forgotten to include an important property.  The declarative nature of formal specifications gives the synthesis process the liberty to resolve unspecified behavior in an arbitrary way. This may result in implementations that satisfy the specification, yet still behave very differently from  the developer's expectations. 

In this paper, we propose a new synthesis process that, while still fully automatic, provides the user with an explanation for the decisions made by the synthesis algorithm. The \emph{explainable synthesis} process builds the implementation incrementally, starting with a small subset of the desired properties and then adding more properties, one at a time. In each step, the algorithm presents an implementation that satisfies the currently considered specification and explains the changes that were required from the previous implementation in order to accomodate the additional property. Such an explanation consists of a counterexample trace, which demonstrates that the previous implementation violated the property, and a transformation that minimally repairs the problem.

As an example, consider the specification of a two-client arbiter in linear-time temporal logic (LTL) shown in \Cref{fig:arbiterSpec}. The specification consists of three properties:  $\varphi_\textit{mutex},\varphi_\textit{fairness}$ and $\varphi_\textit{non-spurious}$, requiring \emph{mutual exclusion}, i.e., there is at most one grant at a time, \emph{fairness}, i.e., every request is eventually granted, and \emph{non-spuriousness}, i.e., a grant is only given upon receiving a request\footnote{The sub-formula $r_i \LTLrelease \neg g_i$  states that initially no grant is given to client $i$ as long as no request is received from this client. After that, the formula $\LTLsquare (g_i \to r_i \lor (\bigcirc (r_i \mathcal{R} \neg g_i)))$ ensures that a grant is active only if the current request is still active, otherwise, and from this point on, no grants are given as long as no new request is received.}. Let us suppose that our synthesis algorithm has already produced 
 the transition system shown in \Cref{fig:roundRobin} for the partial specification $\varphi_\textit{mutex} \wedge \varphi_\textit{fairness}$. This solution does not satisfy $\varphi_\text{non-spurious}$.
  To repair the transition system, the synthesis algorithm carries out the transformations depicted in Figures \ref{fig:transformation1} to \ref{fig:transformation5}. The transformations include a label change in the initial state and the redirection of five transitions. The last four redirections require the expansion of the transition system to two new states $t_2$ and $t_3$. The synthesis algorithm justifies the transformations with counterexamples, depicted in red in Figures \ref{fig:transformation1} to~\ref{fig:transformation4}.
  
  The algorithm justifies the first two transformations, (1) changing the label in the initial state to $\emptyset$ as depicted in \Cref{fig:transformation1} and (2) redirecting the transition  $(t_0,\emptyset,t_1)$ to $(t_0,\emptyset,t_0)$, as shown in \Cref{fig:transformation2}, by a path in the transition system that violates $\varphi_\textit{non-spurious}$, namely the path that starts with  transition $(t_0,\emptyset,t_1)$. 
Changing the label of the initial state causes, however, a violation of $\varphi_\textit{fairness}$, because no grant is given to client 0. This justifies giving access to a new state $t_2$, as shown in \Cref{fig:transformation3} and redirecting the transition with $\{r_0\}$ from $t_0$ to $t_2$.  
The third transformation, leading to \Cref{fig:transformation4}, is justified by the counterexample that, when both clients send a request at the same time, then only client 1 would be given access. 
Finally, the last two transformations, redirecting $(t_1,\{r_0\},t_0)$ to $(t_1,\{r_0\},t_3)$ and $(t_1,\{r_0,r_1\},t_0)$ to $(t_1,\{r_0,r_1\},t_3)$, are justified by the counterexample that if both clients alternate between sending a request then client 0 will not get a grant. This final transformation results in the transition system shown in \Cref{fig:transformation5}, which satisfies all three properties from \Cref{fig:arbiterSpec}. 

\begin{figure}[t!]
\centering
\begin{minipage}[c]{\textwidth}
\begin{subfigure}[c]{\textwidth}	
\begin{align*}
	\varphi_\textit{mutex} & := \hspace{0.3cm} \LTLsquare (\neg g_0 \vee \neg g_1)\\
	\varphi_\textit{fairness} &:= \bigwedge \limits_{i\in \{0,1\}}\LTLsquare(r_i \rightarrow \LTLdiamond g_i)\\
	\varphi_\textit{non-spurious} & := \bigwedge_{i \in \{0,1\}} ((r_i \mathcal{R} \neg g_i) \land \LTLsquare (g_i \to r_i \lor (\bigcirc (r_i \mathcal{R} \neg g_i))))
\end{align*}
\caption[]{Specification of a two-client arbiter}
\label{fig:arbiterSpec}
\end{subfigure}
\end{minipage}

\begin{minipage}[b]{0.33\textwidth}
	\begin{subfigure}[t]{0.9\textwidth}
	\scalebox{0.65}[0.65]{
	\begin{tikzpicture}[initial text=$$]
            	\node[state] (0) {$t_0$};
            	\path[draw,->](-.5,.5) -- (0);
            	\node (l0) [left= 0 of 0] {$\{g_0\}$};
            	\node[state] (1) [right = 3 of 0] {$t_1$};
            	\node (l1) [above = 0 of 1] {$\{g_1\}$};

            	\path[->]
            	(0) edge[bend left =10, draw = red] node[above] {\textcolor{red}{$\emptyset$}$, \{r_0\}, \{r_1\}, \{r_0, r_1\}$} (1);
             	\path[->](1) edge[bend left = 10, draw = red] node[below] {\textcolor{red}{$\emptyset$}$, \{r_0\}, \{r_1\}, \{r_0, r_1\}$} (0);
	\end{tikzpicture}
	}
	\caption{An implementation for specification $\varphi_\textit{mutex} \wedge \varphi_\textit{fairness}$}
	\label{fig:roundRobin}
	\end{subfigure}
\end{minipage}
\hspace{0.1cm}
\begin{minipage}[b]{0.33\textwidth}
	\begin{subfigure}[t]{0.9\textwidth}
	\scalebox{0.65}[0.65]{
	\begin{tikzpicture}[initial text=$$]
            	\node[state] (0) {$t_0$};
            	\path[draw,->](-.5,.5) -- (0);
            	\node (l0) [left= 0 of 0] {$\emptyset$};
            	\node[state] (1) [right = 3 of 0] {$t_1$};
            	\node (l1) [above= 0 of 1] {$\{g_1\}$};

            	\path[->]
            	(0) edge[bend left =10, draw = red] node[above] {\textcolor{red}{$\emptyset$}$, \{r_0\}, \{r_1\}, \{r_0, r_1\}$} (1)
             	(1) edge[bend left = 10, draw = red] node[below] {\textcolor{red}{$\emptyset$}$, \{r_0\}, \{r_1\}, \{r_0, r_1\}$} (0);
            \end{tikzpicture}
	}
	\caption{$\Delta_1$: Changing the label of the initial state}
	\label{fig:transformation1}
	\end{subfigure}
\end{minipage}
%
\begin{minipage}[b]{0.3\textwidth}
	\begin{subfigure}[t]{0.9\textwidth}
	\scalebox{0.65}[0.65]{
\begin{tikzpicture}[initial text=$$]
            	\node[state] (0) {$t_0$};
            	\path[draw,->](-.5,.5) -- (0);
            	\node (l0) [left= 0 of 0] {$\emptyset$};
            	\node[state] (1) [right = 3 of 0] {$t_1$};
            	\node (l1) [above= 0 of 1] {$\{g_1\}$};

            	\path[->]
            	(0) edge[loop above] node[above] {$\emptyset$} (0)
            	(0) edge[bend left =10, draw = red] node[above] {\textcolor{red}{$\{r_0\}$}$, \{r_1\}, \{r_0, r_1\}$} (1)
             	(1) edge[bend left = 10, draw = red] node[below] {$\emptyset, $\textcolor{red}{$\{r_0\}$}$, \{r_1\}, \{r_0, r_1\}$} (0);
            \end{tikzpicture}
	}
	\caption{$\Delta_2$: Redirecting the transition $(t_0,\emptyset,t_1)$ to $(t_0,\emptyset,t_0)$}
	\label{fig:transformation2}
	\end{subfigure}
\end{minipage}

\vskip 0.4cm
\begin{minipage}[b]{0.33\textwidth}
	\begin{subfigure}[t]{0.9\textwidth}
	\scalebox{0.65}[0.65]{
	\begin{tikzpicture}[initial text=$$]
            	\node[state] (0) {$t_0$};
            	\path[draw,->](-.5,.5) -- (0);
            	\node (l0) [left= 0 of 0] {$\emptyset$};
            	\node[state] (1) [right = 3 of 0] {$t_1$};
            	\node (l1) [above= 0 of 1] {$\{g_1\}$};
            	\node[state] (2) [below = 2 of 0] {$t_2$};
            	\node (l2) [below= 0 of 2] {$\{g_0\}$};

            	\path[->]
            	(0) edge[loop above] node[above] {$\emptyset$} (0)
            	(0) edge[bend left =10, draw = red] node[above] {$\{r_1\},$\textcolor{red}{$\{r_0, r_1\}$}} (1)
            	(0) edge[bend right = 15] node[left] {$\{r_0\}$} (2)

             	(1) edge[bend left = 10, draw = red] node[below] {\textcolor{red}{$\emptyset$}$,\{r_0\}, \{r_1\}, \{r_0, r_1\}$} (0)
            	(2) edge[bend right = 15] node[right] {$\emptyset, \{r_0\}$} (0)
            	(2) edge[bend right = 45] node[below, sloped] {$\{r_1\}, \{r_0, r_1\}$} (1);
    \end{tikzpicture}
	}
	\caption{$\Delta_3$: Redirecting the transition $(t_0,\{r_0\},t_1)$ to $(t_0,\{r_0\},t_2)$}
	\label{fig:transformation3}
	\end{subfigure}
\end{minipage}
\begin{minipage}[b]{0.33\textwidth}
	\begin{subfigure}[t]{0.9\textwidth}
	\scalebox{0.65}[0.65]{
		\begin{tikzpicture}[initial text=$$]
            	\node[ state] (0) {$t_0$};
            	\path[draw,->](-.5,.5) -- (0);
            	\node (l0) [left= 0 of 0] {$\emptyset$};
            	\node[state] (1) [right = 3 of 0] {$t_1$};
            	\node (l1) [above= 0 of 1] {$\{g_1\}$};
            	\node[state] (2) [below = 2 of 0] {$t_2$};
            	\node (l2) [below= 0 of 2] {$\{g_0\}$};
				\node[state] (3) [right = 3 of 2] {$t_3$};
            	\node (l3) [below= 0 of 3] {$\{g_0\}$};

            	\path[->]
            	(0) edge[loop above] node[above] {$\emptyset$} (0)
            	(0) edge[bend left =10, draw = red] node[above] {\textcolor{red}{$\{r_1\}$}} (1)
            	(0) edge[bend right = 15] node[left] {$\{r_0\}$} (2)
				(0) edge node[above, sloped, pos = 0.5] {$\{r_0, r_1\}$} (3)
             	(1) edge[bend left = 10, draw = red] node[below, pos = 0.39] {$\emptyset,$\textcolor{red}{$\{r_0\}$}$, \{r_1\}, \{r_0, r_1\}$} (0)
            	(2) edge[bend right = 15] node[right] {$\emptyset, \{r_0\}$} (0)
            	(2) edge[bend right = 45] node[below, sloped, pos = 0.3] {$\{r_1\}, \{r_0, r_1\}$} (1)
            	(3) edge[bend right = 15] node[below =.1, sloped] {\makecell[c]{$\emptyset, \{r_0\}$, $\{r_1\}, \{r_0, r_1\}$}} (1);
            \end{tikzpicture}
	}
	\caption{$\Delta_4$: Redirecting the transition $(t_0,\{r_0,r_1\},t_1)$ to $(t_0,\{r_0,r_1\},t_3)$}
	\label{fig:transformation4}
	\end{subfigure}
\end{minipage}
\begin{minipage}[b]{0.3\textwidth}
	\begin{subfigure}[t]{\textwidth}
	\scalebox{0.65}[0.65]{
		\begin{tikzpicture}[initial text=$$]
            	\node[state] (0) {$t_0$};
            	\path[draw,->](-.5,.5) -- (0);
            	\node (l0) [left= 0 of 0] {$\emptyset$};
            	\node[state] (1) [right = 3 of 0] {$t_1$};
            	\node (l1) [above= 0 of 1] {$\{g_1\}$};
            	\node[state] (2) [below = 2 of 0] {$t_2$};
            	\node (l2) [below= 0 of 2] {$\{g_0\}$};
				\node[state] (3) [right = 3 of 2] {$t_3$};
            	\node (l3) [below= 0 of 3] {$\{g_0\}$};

            	\path[->]
            	(0) edge[loop above] node[above] {$\emptyset$} (0)
            	(0) edge[bend left =10] node[above] {$\{r_1\}$} (1)
            	(0) edge[bend right = 15] node[left] {$\{r_0\}$} (2)
            	(0) edge node[above, sloped, pos = 0.3] {$\{r_0, r_1\}$} (3)
             	(1) edge[bend left = 10] node[below] {$\emptyset, \{r_1\}$} (0)
            	(1) edge node[below, sloped, pos = 0.23] {$\{r_0\}, \{r_0, r_1\}$} (2)
            	(2) edge[bend right = 15] node[right] {$\emptyset, \{r_0\}$} (0)
            	(2) edge[bend right = 45] node[below, sloped, pos = 0.3] {$\{r_1\}, \{r_0, r_1\}$} (1)
            	(3) edge[bend right = 15] node[below=0.1, sloped] {\makecell[c]{$\emptyset, \{r_0\}$,$\{r_1\}, \{r_0, r_1\}$}} (1);
            \end{tikzpicture}
	}
	\caption{$\Delta_5$: Redirecting the transitions $(t_1,v,t_0)$ to $(t_1,v,t_3)$ for $v \in \{\{r_0, r_1\},\{r_0\}\}$}
	\label{fig:transformation5}
	\end{subfigure}
\end{minipage}
\caption{Using explainable reactive synthesis to synthesize an implementation for a two-client arbiter. Clients request access to the shared resource via the signals $r_0$ and $r_1$. Requests are granted by the arbiter via the signals $g_0$ and $g_1$.}	
\label{fig:MotivatingExample}
\end{figure}
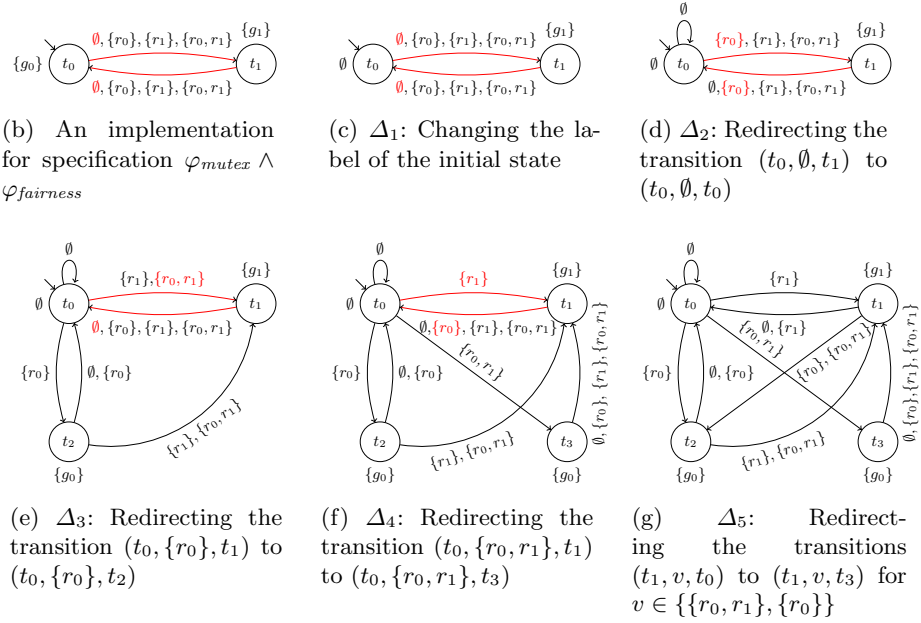

We implement the explainable synthesis approach in the setting of \emph{bounded synthesis}~\cite{BoundedSynthesis,encodingsOfBoundedSynthesis}. Bounded synthesis finds a solution that is minimal in the number of states; this generalizes here to a solution that is obtained from the previous
solution with a minimal number of transformations. Like bounded synthesis, we use a SAT encoding to find the transition system, with additional constraints on the type and number of transformations. As explained above, the transformations involve a change of a state label or a redirection of a transition. Within the given budget of states, new states are accessed by redirecting transitions to these states. In the example in \Cref{fig:MotivatingExample}, a budget of four states is fixed and initially unreachable states, such as $t_2$ and $t_3$, are accessed by redirecting transitions to them as done in \Cref{fig:transformation3}  and \Cref{fig:transformation4}.
For the construction of explanations, we use bounded model checking \cite{BoundedSynthesis}.  In this way, both the repair and the explanation can be ensured to be minimal. We evaluate our approach on a series of examples, including benchmarks from the SYNTCOMP competition \cite{DBLP:journals/corr/abs-1904-07736}.  

\subsection*{Related Work}
The importance of incremental algorithms for solving the reactive synthesis problem has quickly manifested itself in the research community after the introduction of the problem. 
 By decomposing a synthesis problem into smaller instances and combining the results of these instances to a solution for the full problem, the hope is to provide scalable algorithms for solving  the in general difficult problem~\cite{lazySynthesis,safralessSynthesis,refinement1,refinement2,refinement3,termite}. 
 For example, for a set of system specifications, one can construct implementations for the individual specifications and construct an implementation for the full specification by composing the results for the smaller specifications~\cite{safralessSynthesis}. To check the realizability of a specification, one can check the realizability of  gradually larger parts of the specification~\cite{termite}.  
Refinement-based synthesis algorithms incrementally construct implementations starting with an abstraction that is gradually refined with respect to a given partial order that guarantees correctness \cite{lazySynthesis,refinement1,refinement2,refinement3}.
The key difference between our approach and the incremental approaches mentioned above is the  underlying repair process. The advantage of program repair is that it  constructs an implementation that is close to the original erroneous implementation. In our approach, this makes it possible to derive explanations that justify the repairs  applied to the previous implementation. Other repair problems for temporal logics have  previously been considered in \cite{programrepairgame,programRepair}. In \cite{programrepairgame}, an expression or a left-hand side of an assignment is assumed to be erroneous and replaced by one that satisfies the specification. In \cite{programRepair}, the repair removes edges from the transition system. 
By contrast, our repair algorithm changes labels of  states and redirects  transitions.  

A completely different  approach to make synthesized implementation more understandable is by posing conditions on the structure of synthesized implementations~\cite{crest17,outputsensitive}. Bounded synthesis \cite{BoundedSynthesis} allows us to bound the size of the constructed implementation. Bounded cycle synthesis \cite{BoundedCycle} additionally bounds the number of cycles in the implementation. Skeleton synthesis~\cite{skeletons} establishes the relations between the specification and the synthesized implementation to clarify which parts of the implementation are determined by the specification and which ones where chosen by the synthesis process. 

\section{Preliminaries}\label{sec:Preliminaries}

\paragraph{Linear-time Temporal Logic:} As specification language, we use Linear-Time Temporal Logic (LTL)~\cite{LTL}, with the usual temporal operators Next $\LTLcircle$, Until $\LTLuntil$ and the derived operators Release $\LTLrelease$, which is the dual operator of $\LTLuntil$, Eventually $\LTLdiamond$ and Globally $\LTLsquare$. Informally, the Release operator $\varphi_1 \LTLrelease \varphi_2$ says that $\varphi_2$ holds in every step until $\varphi_1$ \textit{releases} this condition.
LTL formulas defining specifications for reactive systems are defined over a set of atomic propositions $AP = I \cup O$, which is partitioned into a set $I$ of input propositions and a set $O$ of output propositions.
We denote the satisfaction of an LTL formula $\varphi$ by an infinite sequence $\sigma \colon \mathbb{N} \to 2^{AP}$ of valuations of the atomic propositions by $\sigma \vDash \varphi$. For an LTL formula $\varphi$ we define the language $\mathcal{L}(\varphi)$ by the set $\{\sigma \in (\mathbb{N} \to 2^{AP}) \; | \; \sigma \vDash \varphi \}$.

\paragraph{Implementations:} We represent implementations as \emph{labeled transition systems}. For a given finite set $\Upsilon$ of directions and a finite set $\Sigma$ of labels, a
 $\Sigma$-labeled $\Upsilon$-transition system is a tuple $\mathcal{T} = (T, t_0, \tau, o)$, consisting of a finite set of states $T$, an initial state $t_0 \in T$, a transition function $\tau \colon T \times \Upsilon \to T$, and a labeling function $o \colon T \to \Sigma$.
 For a set $I$ of input propositions and a set $O$ of output propositions, we represent reactive systems as $2^O$-labeled $2^I$-transition systems.
For reactive systems, a \emph{path} in $\mathcal{T}$ is a sequence $\pi \in \mathbb{N} \to T \times 2^I$ of states and directions that follows the transition function, i.e., for all $i \in \mathbb{N}$, if $\pi(i) = (t_i, e_i)$ and $\pi(i+1) = (t_{i+1}, e_{i+1})$, then $t_{i+1} = \tau(t_i, e_i)$. We call a path initial if it starts with the initial state: $\pi(0) = (t_0, e)$ for some $e \in 2^I$. We denote the set of initial paths of $\mathcal{T}$ by $\mathit{Paths}(\mathcal{T})$. For a path $\pi \in \mathit{Paths}(\mathcal{T})$, we denote the sequence $\sigma_\pi \colon i \mapsto o(\pi(i))$, where $o(t,e) = (o(t) \cup e)$ by the \emph{trace} of $\pi$. We call the set of traces of the paths of a transition system $\mathcal{T}$ the language of $\mathcal{T}$, denoted by $\mathcal{L}(\mathcal{T})$.
	
	For a given finite sequence  $v^* \in (2^I)^*$, we denote the transitions sequence where we reach a state $t'$ from state $t$ after applying the transition function $\tau $ for every letter in $v^*$ starting in $t$ by $\tau^*(t, v^*) = t'$.
	The size of a transition system is the size of its set of states, which we denote by $|\mathcal T|$.

	For a set of atomic propositions $AP = I \cup O$, we say that a $2^O$-labeled $2^I$-transition system $\mathcal{T}$ satisfies an LTL formula $\varphi$, if and only if $\mathcal{L}(\mathcal{T}) \subseteq \mathcal{L}(\varphi)$, i.e., every trace of $\mathcal{T}$ satisfies $\varphi$. In this case, we call $\mathcal{T}$ a model of $\varphi$.

\section{Minimal Repairs and Explanations}\label{sec:FormalStatements}
In this section, we lay the foundation for explainable reactive synthesis. We formally define the transformations that are performed by our repair algorithm and determine the complexity of finding a minimal repair, i.e., a repair with the fewest number of transformations, with respect to a given transition system and an LTL specification and show how repairs can be explained by counterexamples that justify the repair. 

\subsection{Generating Minimal Repairs}

For a $2^O$-labeled $2^I$-transition system $\mathcal{T} = (T, t_0, \tau, o)$, an \emph{operation} $\Delta$ is either a \emph{change of a state labeling} or a \emph{redirection of a transition} in $\mathcal{T}$. We denote the transition system $\mathcal{T'}$ that results from applying an operation $\Delta$ to the transition system $\mathcal{T}$ by $\mathcal{T}'= \mathit{apply}(\mathcal{T}, \Delta)$.

A state labeling change is denoted by a tuple $\Delta_\text{label} = (t, v)$, where $t \in T$ and $v \in 2^O$ defines the new output $v$ of state $t$. The transition system $\mathcal T'= \mathit{apply}(\mathcal T,\Delta_\text{label})$ is defined by $\mathcal T'= (T, t_0, \tau, o')$, where $o'(t) = v$ and  $o'(t') = o(t')$ for all $t' \in T$ with $t' \neq t$.

A transition redirection is denoted by a tuple $\Delta_\text{transition} = (t, t', V)$, where $t,t' \in T$ and $V \subseteq {2^I}$. For a  transition redirection operation $\Delta_\text{transition} = (t, t', V)$, the transition system  $\mathcal T' = \mathit{apply}(\mathcal T, \Delta_\text{transition}) $ is defined by $ \mathcal T'= (T, t_0, \tau', o)$, with $\tau'(t,v) = t'$ for $v \in V$ and $\tau'(t,v) = \tau(t,v)$ for $v \notin V$. For $t'' \neq t$ and $v \in 2^I$, $\tau'(t'',v) = \tau(t'',v)$.

A finite set of operations $\xi$ is called a \emph{transformation}. A transformation $\xi$ is \emph{consistent} if there is no transition system $\mathcal{T}$ and $\Delta_1, \Delta_2 \in \xi$ such that $\mathit{apply}(\mathit{apply}(\mathcal{T}, \Delta_1),\Delta_2) \neq \mathit{apply}(\mathit{apply}(\mathcal{T}, \Delta_2),\Delta_1)$, i.e.\ the resulting transition system does not differ depending on the order in which operations are applied.
For a consistent transformation $\xi$, we denote the $2^O$-labeled $2^I$-transition system $\mathcal{T'}$ that we reach after applying every operation in $\xi$ starting with a $2^O$-labeled $2^I$-transition system $\mathcal{T}$ by $\mathcal T'=\mathit{apply}^*(\mathcal{T}, \xi)$. 

Note that there is no operation which explicitly adds a new state. In the example in \Cref{fig:MotivatingExample}, we assume a fixed number of available states (some that might be unreachable in the initial transition system). We reach new states by using a transition redirection operation to these states.

\begin{definition}[Minimal Repair]
	For an \emph{LTL}-formula $\varphi$ over $AP = I \cup O$ and a $2^O$-labeled $2^I$-transition system $\mathcal{T}$, a consistent transformation $\xi$ is a \emph{repair} for $\mathcal T$ and $\varphi$ if $\mathit{apply}^*(\mathcal{T}, \xi) \vDash \varphi$. A repair $\xi$ is \emph{minimal} if  there is no repair $\xi'$ with $|\xi'| < |\xi|$.
\end{definition}

\begin{example}\label{ex:transformations}
	The arbiter $\mathit{Arb}_1$ in \Cref{fig:transformation1} can be obtained from the round-robin arbiter $\mathit{Arb}_0$, shown in \Cref{fig:roundRobin}, by applying $\Delta_\text{label} = (t_0, \emptyset)$, i.e.\ $\mathit{Arb}_1 = \mathit{apply}(\mathit{Arb}_0, \Delta_\text{label})$. Arbiter $\mathit{Arb}_3$, depicted in \Cref{fig:transformation3} is obtained from $\mathit{Arb}_1$ with the transformation $\xi_1 = \{\Delta_\text{transition1},\Delta_\text{transition2}\}$ where $\Delta_\text{transition1} = (t_0, t_0, \{\emptyset\})$ and $\Delta_\text{transition2} = (t_0, t_2, \{\{r_0\}\})$ such that $\mathit{apply}^*(\mathit{Arb}_1, \xi_1) = \mathit{Arb}_3$. A minimal repair for $\mathit{Arb}_0$ and $\varphi_\textit{mutex} \land \varphi_\textit{fairness} \land \varphi_\textit{non-spurious}$, defined in \Cref{sec:Introduction}, is $\xi_2 = \{\Delta_\text{label}, \Delta_\text{transition1}, \Delta_\text{transition2}, \Delta_\text{transition3}, \Delta_\text{transition4}\}$ with $\Delta_\text{transition3} = (t_0, t_3, \{\{r_0, r_1\}\})$ and $\Delta_\text{transition} = (t_1, t_2, \{\{r_0\}, \{r_0, r_1\}\})$. The resulting full arbiter $\mathit{Arb}_5$ is depicted in \Cref{fig:transformation5}, i.e.\ $\mathit{apply}^*(\mathit{Arb}_0, \xi_2) = \mathit{Arb}_5$.
\end{example}

We are interested in finding minimal repairs. The \emph{minimal repair synthesis problem} is defined as follows.

\begin{probdef}[Minimal Repair Synthesis]
	Let $\varphi$ be an \emph{LTL}-formula over a set of atomic propositions $AP = I \cup O$ and $\mathcal{T}$ be a $2^O$-labeled $2^I$-transition system. Find a minimal repair for $\varphi$ and $\mathcal{T}$?
\end{probdef} 

In the next lemma, we show that for a fixed number of operations, the problem of checking if there is a repair is polynomial in the size of the transition system and exponential in the number of operations. For a small number of operations, finding a repair is cheaper than synthesizing a new system, which is 2EXPTIME-complete in the size of the specification \cite{article}. 

\begin{lemma}\label{theorem:decision_minimal_transformation_fixed_state_space}
	For an \textup{LTL}-formula $\varphi$, a $2^O$-labeled $2^I$-transition system $\mathcal{T}$, and a bound $k$, deciding whether there exists a repair $\xi$ for $\mathcal T$ and $\varphi$ with $|\xi|=k$ can be done in time polynomial in the size of $\mathcal T$, exponential in $k$, and space polynomial in the length of $\varphi$. 
\end{lemma}
	\begin{proof}
		Checking for a transformation $\xi$ if $\mathit{apply}^*(\mathcal{T}, \xi) \vDash \varphi$ is \textup{PSPACE}-complete \cite{clarke}.
		There are $|\mathcal{T}| \cdot 2^{|O|}$ different state labeling operations and $|\mathcal{T}|^2 \cdot 2^{|I|}$ transition redirections. Thus, the number of transformations $\xi$ with $|\xi| = k$ is bounded by $\mathcal{O}((|\mathcal{T}|^2)^k)$. Hence, deciding the existence of such a repair is polynomial in $|\mathcal{T}|$ and exponential in $k$. \qed
	\end{proof}

The size of a minimal repair is bounded by a polynomial in the size of the transition system under scrutiny. Thus, the minimal repair synthesis problem can be solved in time at most exponential in $|\mathcal T|$. In most cases, we are interested in small repairs resulting in complexities  that are polynomial in $|\mathcal T|$. 
\begin{theorem}
	For an \textup{LTL}-formula $\varphi$, a $2^O$-labeled $2^I$-transition system $\mathcal{T}$, finding a minimal repair for $\mathcal T$ and $\varphi$ can be done in time exponential in the size of $\mathcal T$,  and space polynomial in the length of $\varphi$. 
\end{theorem}

\subsection{Generating Explanations}

For an LTL-formula $\varphi$ over $AP = I \cup O$, a transformation $\xi$ for a $2^O$-labeled $2^I$-system $\mathcal{T}$ is \emph{justified} by a counterexample $\sigma$ if $\sigma \nvDash \varphi$, $ \sigma \in \mathcal{L}(\mathcal{T})$ and $\sigma \notin \mathcal{L}(\mathit{apply}^*(\mathcal{T}, \xi))$. We call $\sigma$ a \emph{justification} for $\xi$. A transformation $\xi$ is called \emph{justifiable} if there exists a justification $\sigma$ for $\xi$.

 A transformation $\xi$ for $\mathcal{T}$ and $\varphi$ is \emph{minimally} justified by $\sigma$ if $\xi$ is justified by $\sigma$ and there is no $\xi' \subset \xi$ where $\sigma$ is a justification for $\xi'$. If a transformation $\xi$ is minimally justified by a counterexample $\sigma$, we call $\sigma$ a \emph{minimal} justification.
 
\begin{definition}[Explanation]
	For an \emph{LTL}-formula $\varphi$ over $AP = I \cup O$, an initial $2^O$-labeled $2^I$-transition system $\mathcal{T}$, and a minimal repair $\xi$, an \emph{explanation} $\mathit{ex}$ is defined as a sequence of pairs of transformations and counterexamples.  For an explanation $\mathit{ex} = (\xi_1, \sigma_1), \dots, (\xi_n, \sigma_n)$, it holds that all transformations $\xi_1,\dots,\xi_n$ are disjoint, $\xi = \bigcup_{1 \leq i \leq n } \xi_i$, and each transformation $\xi_i$ with $1 \leq i \leq n$ is minimally justified by $\sigma_i$ for $\mathit{apply}^*(\mathcal{T}, \bigcup_{1 \leq j < i} \xi_j)$ and $\varphi$.
\end{definition}

\begin{example}
	Let $\varphi_1 = g \LTLand \LTLcircle \neg g$ over $I = \{r\}$ and $O = \{g\}$ and consider the $2^O$-labeled $2^I$-transition system $\mathcal{T}$ with states $\{t_0,t_1\}$, depicted in \Cref{fig:example}.
	\makeatletter
	\let\par\@@par
	\par\parshape0
	\everypar{}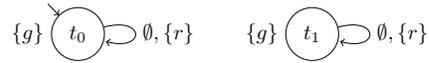
\begin{wrapfigure}{r}{.5\textwidth} 
		\centering
		\scalebox{0.8}{
		\begin{tikzpicture}[initial text=$$]
            	\node[state] (0) {$t_0$};
            	\path[draw,->](-.5,.5) -- (0);
				\node[state] (1) [right = 3 of 0] {$t_1$};
				\node (label0) [left = 0 of 0] {$\{g\}$};
				\node (label1) [left = 0 of 1] {$\{g\}$};
				\path[->]
				(0) edge[loop right] node[right] {$\emptyset, \{r\}$} (0)
				(1) edge[loop right] node[right] {$\emptyset, \{r\}$} (1);
		\end{tikzpicture}
		}
		\caption{A transition system over $I = \{r\}$ and $O = \{g\}$ that is not a model of $\varphi_1$.}
		\label{fig:example}
	\end{wrapfigure}
	\noindent For $\mathcal{T}$ and $\varphi_1$, the transformation $\xi$ with $\xi = \{\Delta_\text{transition}\}$ where $\Delta_\text{transition} = (t_0, t_1, \{\{g\}, \emptyset\})$, is not justifiable because $\mathcal{L}(\mathcal{T}) = \mathcal{L}(\mathit{apply}^*(\mathcal{T}, \xi))$. For our running example, introduced in \Cref{sec:Introduction}, the transformation $\xi_1 = \{\Delta_\text{label}\}$ that is defined in \Cref{ex:transformations}, is justifiable for the round-robin arbiter $\mathit{Arb}_0$ and $\varphi_\text{mutex} \land \varphi_\text{fairness}\land \varphi_\text{non-spurious}$. It is justified by the counterexample $\sigma_1 = (\{g_0\} \cup \emptyset, \{g_1\} \cup \emptyset)^\omega$, indicated by the red arrows in \Cref{fig:roundRobin}. Further, $\sigma_1$ is a minimal justification. The transformation $\xi_2 = \{\Delta_\text{label}, \Delta_\text{transition1}\}$ for $\mathit{Arb}_0$ is not minimally justified by $\sigma_1$ as $\sigma_1$ is a justification for $\xi_1$ and $\xi_1 \subset \xi_2$. An explanation $\mathit{ex}$ for $\mathit{Arb}_0$, the LTL-formula $\varphi_\text{mutex} \land \varphi_\text{fairness}\land \varphi_\text{non-spurious}$ and the minimal repair $\xi_3 = \{\Delta_\text{label}, \Delta_\text{transition1}, \Delta_\text{transition2}, \Delta_\text{transition3}, \Delta_\text{transition4}\}$ is $\mathit{ex} = (\Delta_\text{label}, \sigma_1),(\Delta_\text{transition1},\sigma_2), (\Delta_\text{transition2}, \sigma_3), (\Delta_\text{transition3}, \sigma_4), (\Delta_\text{transition4}, \sigma_5)$\\ with $\sigma_2 = (\emptyset \cup \emptyset, \{g_1\} \cup \emptyset)^\omega$, $\sigma_3 = (\emptyset \cup \{r_0\}, \{g_1\} \cup \{r_0\})^\omega$, $\sigma_4 = (\emptyset \cup \{r_0, r_1\}, \{g_1\} \cup \emptyset)^\omega$ and $\sigma_5 = (\emptyset \cup \{r_1\}, \{g_1\} \cup \{r_0\})^\omega$. The different justifications are indicated in the subfigures of \Cref{fig:MotivatingExample}.
\end{example}

In the next theorem, we show that there  exists an explanation for every minimal repair.

\begin{theorem}\label{theorem:exists_expl_for_every_minimal_tf}
	For every minimal repair $\xi$ for an \textup{LTL}-formula $\varphi$ over $AP = I \cup O$ and a $2^O$-labeled $2^I$-transition system $\mathcal{T}$, there exists an explanation.
\end{theorem}
\begin{proof}
		Let $\xi = \{\Delta_1, \dots, \Delta_n\}$ be a minimal repair for the $LTL$-formula $\varphi$ and the transition system $\mathcal{T}$. An explanation $\mathit{ex}$ can be constructed as follows. Let $\sigma \in \mathcal{L}(\mathcal{T})$ with $\sigma \nvDash \varphi$. Since $\xi$ is a minimal repair, $\sigma \notin \mathcal{L}(\mathit{apply}^*(\mathcal{T}, \xi))$. The smallest subset $\xi' \subseteq \xi$ with $\sigma \notin \mathcal{L}(\mathit{apply}^*(\mathcal{T}, \xi'))$ is minimally justified by $\sigma$. Thus $(\xi',\sigma)$ is the first element of the explanation $\mathit{ex}$. For the remaining operations in $\xi \backslash \xi'$, we proceed analogously. The counterexample $\sigma$ is now determined for $\mathit{apply}^*(\mathcal{T}, \xi')$. The construction is finished if either every transformation is minimally justified by a counterexample and there is no operation left or there is no justification for a transformation which clearly contradicts that $\xi$ is a minimal repair. Hence, $\mathit{ex}$ is an explanation for $\xi$. \qed
\end{proof}

From the last theorem we know that we can find an explanation for every minimal repair. It is however important to notice that it is not necessarily the case that we can find justifications for each singleton transformation in the repair as shown by the following example. Let $\varphi_2 = \neg g \land \LTLcircle \neg g \land ((\LTLsquare \neg r) \to \LTLcircle \LTLcircle g)$ over $I = \{r\}$, $O = \{g\}$ and consider the $2^O$-labeled $2^I$-transition system $\mathcal{T}$ with the set of states $\{t_0, t_1, t_2\}$, depicted in \Cref{fig:example2}.
	\begin{wrapfigure}{r}{.5\textwidth}
		\centering
		\scalebox{0.8}{
		\begin{tikzpicture}[initial text=$$]
            	\node[state] (0) {$t_0$};
            	\path[draw,->](-.5,.5) -- (0);
				\node[state] (1) [right = 1.5 of 0] {$t_1$};
				\node[state] (2) [right = 1.5 of 1] {$t_2$};
				\node (label0) [below = 0 of 0] {$\emptyset$};
				\node (label1) [below = 0 of 1] {$\emptyset$};
				\node (label2) [below = 0 of 2] {$\{g\}$};
				\path[->]
				(0) edge[loop above] node[above] {$\emptyset$} (0)
				(0) edge node[above] {$\{r\}$} (1)
				(1) edge[loop above] node[above] {$\emptyset$} (1)
				(1) edge node[above] {$\{r\}$} (2)
				(2) edge[loop above] node[above] {$\emptyset, \{r\}$} (2);
		\end{tikzpicture}
	}
	\caption{A transition system over $I = \{r\}$ and $O = \{g\}$ that is not a model of $\varphi_2$.}
	\label{fig:example2}
	\end{wrapfigure}
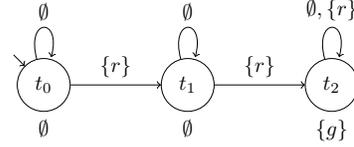
\noindent For $\varphi_2$ and $\mathcal{T}$, the transformation $\xi$ with $\xi = \{\Delta_\text{transition1}, \Delta_\text{transition2}\}$ where $\Delta_\text{transition1} = (t_0, t_1, \{\emptyset\}),$ and $\Delta_\text{transition2} = (t_1, t_2, \{\emptyset\})$, is a minimal repair. The counterexample $\sigma = \emptyset^\omega$ is the only one with $\sigma \in \mathcal{L}(\mathcal{T})$. For an explanation $\mathit{ex} = (\xi_1, \sigma_1),(\xi_2, \sigma_2)$ where $\xi_i$ is a singleton, for all $1 \leq i \leq 2$, either  $\xi_1 = \{\Delta_\text{transition1}\}$ or $\xi_1 = \{\Delta_\text{transition2}\}$. However, in both cases, $\sigma \in \mathcal{L}(\mathit{apply}^*(\mathcal{T}, \xi_1))$. Thus, there are minimal repairs where not every operation can be justified on its own. Furthermore, it should be noted that explanations are not unique as there can exist different justifications for the same transformation, i.e.\ there can exist multiple different explanations for the same minimal repair. For the round-robin arbiter in \Cref{fig:roundRobin} and the specification $\varphi_\textit{mutex} \land \varphi_\textit{fairness}\land \varphi_\textit{non-spurious}$, the transformation $\xi_1 = \{\Delta_\text{label}\}$ can be minimally justified by $(\{g_0\} \cup \emptyset, \{g_1\} \cup \emptyset)^\omega$ and by $(\{g_0\} \cup \{r_1\}, \{g_1\} \cup \emptyset)^\omega$.\\

We refer to the problem of finding an explanation for a minimal repair as the the \emph{explanation synthesis problem}. 
\begin{probdef}[Explanation Synthesis]
Let $\varphi$ be an LTL-formula over $AP = I \cup O$, $\mathcal{T}$ be a $2^O$-labeled $2^I$-transition system and $\xi$ be a minimal repair. Find an explanation $\mathit{ex}$ for $\varphi$, $\mathcal{T}$ and $\xi$. 
\end{probdef}

\section{SAT-based Algorithms for Minimal Repair and Explanation Synthesis}\label{sec:algorithms}

In this section, we present SAT-based algorithms to solve the minimal repair synthesis problem and the explanation synthesis problem. 

\begin{algorithm}[t]
\caption{\textsc{MinimalRepair}($\mathcal{T}, \varphi$)}\label{alg:minimalTf}
\begin{multicols}{2}
\begin{algorithmic}[1]
\State $\mathit{left} \gets 0$
\State $\mathit{right} \gets |\mathcal{T}| + |\mathcal{T}| \cdot |\mathcal{T}|$
\State $\mathit{exist} \gets \mathit{false}$
\While {$\mathit{left} < \mathit{right}$}
\State $k \gets \lfloor \frac{\mathit{left} + \mathit{right}}{2} \rfloor$
\State $(\mathit{found}, \xi) \gets \textsc{Repair}(\mathcal{T}, \varphi, k)$
\If {$\mathit{found}$}
\State $\mathit{right} \gets k-1$
\State $\xi_\textit{min} \gets \xi$
\State $\mathit{exists} \gets \mathit{true}$
\Else 
\State $\mathit{left} \gets k+1$
\EndIf
\EndWhile
\columnbreak
\State $(\mathit{found}', \xi') \gets \textsc{Repair}(\mathcal{T}, \varphi, \mathit{left})$
\If {$!\mathit{exists}$}
\State \Return no minimal repair exists
\EndIf
\If {$\mathit{found}'$}
\State \Return $\xi'$
\Else 
\State \Return $\xi_\textit{min}$
\EndIf
\end{algorithmic}
\end{multicols}
\end{algorithm}

\subsection{Generating Minimal Repairs}\label{sec:alg_for_min_transformation}

The procedure $\textsc{MinimalRepair}(\mathcal{T}, \varphi)$, shown in \Cref{alg:minimalTf}, solves the minimal repair synthesis problem. For a given LTL-formula $\varphi$ over $AP = I \cup O$ and $2^O$-labeled $2^I$-transition system $\mathcal{T}$, \Cref{alg:minimalTf} constructs a minimal repair $\xi$ if one exists. We use binary search to find the minimal number $k$ of required operations. The possible number of operations can be bounded by $|\mathcal{T}|+|\mathcal{T}|\cdot |\mathcal{T}|$ as there are only $|\mathcal{T}|$ state labelings and $|\mathcal{T}| \cdot |\mathcal{T}|$ transition redirects. Checking whether there is a transformation $\xi$ with $|\xi|\leq k$ such that $\mathit{apply}^*(\mathcal{T},\xi) \vDash \varphi$ is done by using the procedure $\textsc{Repair}(\mathcal{T},\varphi,k)$ which is explained next.\\

\begin{figure}[t!]
\begin{align*}
    	\phi_{cost} = &\bigwedge_{0 \leq t,n < |\mathcal{T}|, c \leq k} \mathit{rdTrans}_{t,n,c} \land \mathit{notRdTrans}_{t,n,c} \land \neg cost_{t,n, k+1} \\
    	&\bigwedge_{0 \leq t < |\mathcal{T}|, c \leq k} \mathit{changeLabel}_{t,c} \land \mathit{notChangeLabel}_{t,c} \land \neg cost_{t, |\mathcal{T}|, k+1} \\
    			&\\
    	\mathit{rdTrans}_{t, n, c} &=
    	\begin{cases}
    	\mathit{trans}_{0,0} \to \mathit{cost}_{0,0,1} &\text{if $t = 0 \land n = 0$}\\
    	\mathit{cost}_{t-1,|\mathcal{T}|,c} \land \mathit{trans}_{t,n} \to \mathit{cost}_{t,n,c+1} &\text{if $t > 0 \land n = 0$}\\
    	\mathit{cost}_{t,n-1,c} \land \mathit{trans}_{t,n} \to \mathit{cost}_{t,n,c+1} &\text{if $ n > 0$}\\
    	\end{cases}\\
    	\mathit{notRdTrans}_{t, n, c} &=
    	\begin{cases}
    	\neg \mathit{trans}_{0,0} \to \mathit{cost}_{0,0,0} &\text{if $t = 0 \land n = 0$}\\
    	\mathit{cost}_{t-1,|\mathcal{T}|,c} \land \neg \mathit{trans}_{t,n} \to \mathit{cost}_{t,n,c} &\text{if $t > 0 \land n = 0$}\\
    	\mathit{cost}_{t,n-1,c} \land \neg \mathit{trans}_{t,n} \to \mathit{cost}_{t,n,c} &\text{if $ n > 0$}\\
    	\end{cases}\\
    	&\\
    	\mathit{changeLabel}_{t,c} &= \mathit{cost}_{t,|\mathcal{T}|-1,c} \land \mathit{label}_t \to \mathit{cost}_{t,|\mathcal{T}|,c+1}\\
    	\mathit{notChangeLabel}_{t,c} &= \mathit{cost}_{t,|\mathcal{T}|-1,c} \land \neg \mathit{label}_t \to \mathit{cost}_{t,|\mathcal{T}|,c}\\
    	&\\
      	\mathit{label}_{t} &= \bigvee_{o \in O} 
      		\begin{cases}
      			o'_t &\text{if $o \notin o(t)$}\\
      			\neg o'_t &\text{if $o \in o(t)$}\\
      		\end{cases}\\
      	\mathit{trans}_{t,t'} &= \bigvee_{i \in 2^I}
      		\begin{cases}
      			\tau'_{t,i,t'} &\text{if $\tau(t,i) \neq t'$}\\
      			\bot &\text{if $\tau(t,i) = t'$}\\
      		\end{cases}
    \end{align*}
\caption{The constraint $\phi_\textit{cost}$ ensures that at most $k$ operations are applied.}
\label{fig:constraintSystem}
\end{figure}	

$\textsc{Repair}(\mathcal{T}, \varphi, k):$ To check whether there is a repair $\xi$ for a $2^O$-labeled $2^I$-transition system $\mathcal{T}$ with $k$ operations, we need to ensure that the resulting transition system is a model for $\varphi$, i.e.\ $\mathit{apply}^*(\mathcal{T},\xi) \vDash \varphi$. 
To check the existence of  a transition system $\mathcal{T}'$ with bound $n = |\mathcal{T}'|$ that implements $\varphi$, we use the SAT-based encoding of bounded synthesis \cite{encodingsOfBoundedSynthesis}. Bounded synthesis is a synthesis procedure for LTL-formulas that produces size-optimal transition systems~\cite{BoundedSynthesis}. For a given LTL formula $\varphi$, a universal co-B\"uchi automaton $\mathcal{A}$ that accepts $\mathcal{L}(\varphi)$ is constructed. A transition system $\mathcal{T}$ satisfies $\varphi$ if every path of the run graph, i.e.\ the product of $\mathcal{T}$ and $\mathcal{A}$, visits a rejecting state only finitely often. An annotation function $\lambda$ confirms that this is the case. The bounded synthesis approach constructs a transition system with bound $n$ by solving a constraint system that checks the existence of a transition system $\mathcal{T}$ and a valid annotation function $\lambda$. In the bounded synthesis constraint system for the $2^O$-labeled $2^I$-transition system $\mathcal{T}' = (T, t_0, \tau', o')$, the transition function $\tau'$ is represented by a variable $\tau'_{t,i,t'}$ for every $t,t' \in T$ and $i \in 2^I$.  The variable $\tau'_{t,i,t'}$ is true if and only if $\tau'(t,i) = t'$. The labeling function $o'$ is represented by a variable $o_t'$ for every $o \in O$ and $t \in T$ and it holds that $o_t'$ is true if and only if $o \in o'(t)$. For simplicity, states are represented by natural numbers.

To ensure that the transition system $\mathcal{T}'$ can be obtained with at most $k$ operations from a given transition system $\mathcal{T}= (T, t_0, \tau, o)$, the bounded synthesis encoding is extended with the additional constraint $\phi_{cost}$ shown in \Cref{fig:constraintSystem}.
For states $t,t'$, the constraint $\mathit{trans}_{t,t'}$ holds iff there is a redirected transition from $t$ to $t'$, i.e.\ there exists an $i \in 2^I$ with $\tau'(t,i)=t'$ and $\tau(t,i) \neq t'$. The constraint $label_t$ holds iff the state labeling of state $t$ is changed, i.e.\ $o(t) \neq o'(t)$.
To count the number of applied operations, we use an implicit ordering over all the possible operations: starting from state 0, we first consider all potential transition redirects to states $0, 1, \ldots, |\mathcal{T}|-1$, then the potential state label change of state 0, then the transition redirects from state 1, and so on.   For state $t$ and operation $n$, where $n$ ranges from $0$ to $|\mathcal{T}|$ (where  $n<|\mathcal{T}|$ refers to the transition redirect operation to state $n$ and $n=|\mathcal{T}|$ refers to the state label change operation), the variable $\mathit{cost}_{t,n,c}$ is true if the number of applied operations so far is $c$. This bookkeeping is done by constraints $\mathit{rdTrans}_{t, n, c}$, $\mathit{notRdTrans}_{t, n, c}$, $\mathit{changeLabel}_{t,c}$ and $\mathit{notChangeLabel}_{t,c}$.
Constraints $\mathit{rdTrans}$ and $\mathit{notRdTrans}$ account for the presence and absence, respectively, of transition redirects,
constraints $\mathit{changeLabel}_{t,c}$ and $\mathit{notChangeLabel}_{t,c}$ for the presence and absence of state label changes.
In order to bound the total number of operations by $k$, $\phi_{cost}$ requires that $cost_{t,n, k+1}$ is false for all states $t$ and operations $n$.

In the next theorem, we state the size  of the resulting constraint system, based on the size of the bounded synthesis constraint system given in \cite{encodingsOfBoundedSynthesis}.

	\begin{theorem}\label{theorem:complexity-constraint-system}
		The size of the constraint system is in $\mathcal{O}(nm^2\cdot 2^{|I|} \cdot (|\delta_{q,q'}| + n \log(nm)) + kn^2)$ and the number of variables is in $\mathcal{O}(n(m \log(nm) + 2^{|I|} \cdot (|O|+n)) + kn^2)$, where $n = |\mathcal{T}'|, m = |Q|$ and $k$ the number of allowed operations.
	\end{theorem}

\begin{algorithm}[t]
\caption{\textsc{Explanation}($\mathcal{T}, \varphi, \xi$)}\label{alg:explanation}
\begin{algorithmic}[1]
\State $\mathcal{T}_\textit{old} \gets \mathcal{T}$
\State $\mathit{ex} \gets ()$
\While {$\xi \neq \emptyset$}
\State $\sigma \gets \textsc{BMC}(\mathcal{T}_\textit{old}, \varphi)$
\State $\mathit{minimal} \gets \mathit{false}$
\State $\xi' \gets \xi$
\While {$!\mathit{minimal}$}
\State $\mathit{minimal} \gets \mathit{true}$
\For {$\Delta \in \xi'$}
\State $\mathcal{T}_\textit{new} \gets \mathit{apply}^*(\mathcal{T}_\textit{old}, \xi'\backslash\{\Delta\})$
\If {$\sigma \notin \mathcal{L}(\mathcal{T}_\textit{new})$}
\State $\mathit{minimal} \gets \mathit{false}$
\State $\xi' \gets \xi'\backslash\{\Delta\}$
\EndIf
\EndFor
\EndWhile
\State $\mathcal{T}_\textit{old} \gets \mathit{apply}^*(\mathcal{T}_\textit{old}, \xi')$
\State $\mathit{ex} \gets \mathit{ex}.\textsc{Append}(\xi', \sigma)$
\State $\xi \gets \xi \backslash \xi'$
\EndWhile
\State \Return $\mathit{ex}$
\end{algorithmic}
\end{algorithm}

\subsection{Generating Explanations}\label{sec:alg_for_explanation}

We describe now how we can solve the explanation synthesis problem for a given LTL-formula $\varphi$ over $AP = I \cup O$, a $2^O$-labeled $2^I$-transition system $\mathcal{T}$ and a minimal repair $\xi$. The minimal repair $\xi$ can be obtained from \Cref{alg:minimalTf}. 
The construction of the explanation follows the idea from the proof of \Cref{theorem:exists_expl_for_every_minimal_tf} and is shown in \Cref{alg:explanation}. An explanation $\mathit{ex}$ is a sequence of transformations $\xi_i$ and counterexamples $\sigma_i$ such that every transformation $\xi_i$ can be minimally justified by $\sigma_i$. A counterexample $\sigma$ for the current transition system $\mathcal{T}_{old}$ is obtained by Bounded Model Checking (BMC) \cite{BMC} and is a justification for $\xi$ as $\xi$ is a minimal repair. BMC checks if there is a counterexample of a given bound $n$ that satisfies the negated formula $\neg \varphi$ and is contained in $\mathcal{L}(\mathcal{T})$. 
The constraint system $\phi_\mathcal{T} \land \phi_\textit{loop} \land \llbracket \neg \varphi \rrbracket$ is composed of three components. $\phi_{\mathcal{T}}$ encodes the transition system $\mathcal{T}$, where each state $t \in T$ is represented as a boolean vector. $\phi_\textit{loop}$ ensures the existence of exactly one loop of the counterexample, and the fixpoint formula $\llbracket \neg\varphi \rrbracket$ ensures that the counterexample satisfies the LTL formula. To obtain a minimal justification, we need to ensure that there is no transformation $\xi' \subset \xi$ such that $\sigma$ justifies $\xi'$. As long as there is an operation $\Delta$ such that $\sigma \notin \mathcal{L}(\mathit{apply}^*(\mathcal{T}_{old}, \xi'\backslash\{\Delta\}))$, $\sigma$ is not a minimal justification for $\xi'$. Otherwise $\sigma$ minimally justifies $\xi'$ and $(\xi',\sigma)$ can be appended to the explanation. The algorithm terminates and returns an explanation if $\xi$ is empty, i.e.\ every operation is justified. The presented algorithm solves the BMC-problem at most $|\xi|$-times and the number of checks if a counterexample is contained in the language of a transition system is bounded by $|\xi|^2$. The correctness of \Cref{alg:explanation} is shown in \Cref{theorem:exists_expl_for_every_minimal_tf}.

\section{Experimental Results}\label{sec:Evaluation}

We compare our explainable synthesis approach with \textbf{BoSy} \cite{BoSy}, a traditional synthesis tool, on several benchmarks. After describing the different benchmark families and technical details, we explain the observable results.  

\subsection{Benchmark Families}

The benchmarks families for arbiter, AMBA and load balancer specifications are standard specifications of SYNTCOMP \cite{SYNTCOMP2017}. For the scaling benchmarks only a constant number of operations is needed. The remaining benchmarks synthesize support for different layers of the OSI communication network. 

\begin{itemize}
	\item 
		\textbf{Arbiter:} An \emph{arbiter} is a control unit that manages a shared resource. $\mathit{Arb}_n$ specifies a system to eventually grant every request for each of the $n$ clients and includes mutual exclusion, i.e.\ at most one grant is given at any time. $\mathit{ArbFull}_n$ additionally does not allow spurious grants, i.e.\ there is only given a grant for client $i$ if there is an open request of client $i$. 
	\item
		\textbf{AMBA:} The \emph{ARM Advanced Microcontroller Bus Architecture} (AMBA) is an arbiter allowing additional features like locking the bus for several steps. The specification $\mathit{AMBAEnc}_n$ is used to synthesize the encode component of the decomposed AMBA arbiter with $n$ clients that need to be controlled. $\mathit{AMBAArb}_n$ specifies the arbiter component of a decomposed AMBA arbiter with an $n$-ary bus, and $\mathit{AMBALock}_n$ specifies the lock component.
	\item
		\textbf{Load Balancer:} A \emph{load balancer} distributes a number of jobs to a fixed number of server. $\mathit{LoadFull}_n$ specifies a load balancer with $n$ clients.
	\item
		\textbf{Bit Stuffer:} \emph{Bit stuffing} is a part of the physical layer of the OSI communication network which is responsible for the transmission of bits. Bit stuffing inserts non-information bits into a bit data stream whenever a defined pattern is recognized. $\mathit{BitStuffer}_n$ specifies a system to signal every recurrence of a pattern with length $n$.
	\item
		\textbf{ABP:} The \emph{alternating bit protocol} (ABP) is a standard protocol of the data link layer of the OSI communication network which transmits packets. Basically, in the ABP, the current bits signals which packet has to be transmitted or received. $\mathit{ABPRec}_n$ specifies the ABP with $n$ bits for the receiver and $\mathit{ABPTran}_n$ for the transmitter.
	\item
		\textbf{TCP-Handshake:} A \emph{transmission control protcol} (TCP) supports the transport layer of the OSI communication network which is responsible for the end-to-end delivery of messages. A TCP-handshake starts a secure connection between a client and a server. $TCP_n$ implements a TCP-handshake where $n$ clients can continuously connect with the server.
	\item
		\textbf{Scaling:} The $\mathit{Scaling}_n$ benchmarks specify a system of size $n$. To satisfy the specification $\mathit{Scaling}'_n$ a constant number of operations is sufficient.
\end{itemize}

\begin{table}[t!]
\begin{center}
	\caption{Benchmarking results of BoSy and our explainable synthesis tool}
	\label{table:Benchmarks}
	\scalebox{0.8}[0.8]{
	\begin{tabular}{ c | c | c | c | c | c | c | c | c }
		\hline
		\multicolumn{2}{c|}{Initial} & \multicolumn{2}{c|}{Extended} & Size & Operations & Number & Time in sec. & Time in sec.\\
		Ben. & Size & Ben. & Size & Aut. & chL/rdT & Just. & BoSy & Explainable\\
		\hline
		$\mathit{Arb}_2$ & 2 & $\mathit{Arb}_4$ & 4 & 5 & 0/3 & 3 & 0.348 & 1.356 \\
		$\mathit{Arb}_4$ & 4 & $\mathit{Arb}_5$ & 5 & 6 & 0/2 & 2 & 2.748 & 12.208\\
		$\mathit{Arb}_4$ & 4 & $\mathit{Arb}_6$ & 6 & 7 & 0/3 & 3 & 33.64 & 139.088\\
		$\mathit{Arb}_4$ & 4 & $\mathit{Arb}_8$ & - & 9 & - & - & timeout & timeout\\
		$\mathit{Arb}_2$ & 2 & $\mathit{ArbFull}_2$ & 4 & 6 & 3/7 & 10 & 0.108 & 0.352 \\
		$\mathit{ArbFull}_2$ & 4 & $\mathit{ArbFull}_3$ & 8 & 8 & 1/18 & 19 & 26.14 &  288.168\\
		\hline
		$\mathit{AMBAEnc}_2$ & 2 & $\mathit{AMBAEnc}_4$ & 4 & 6 & 1/11 & 12 & 0.26 & 7.16\\
		$\mathit{AMBAEnc}_4$ & 4 & $\mathit{AMBAEnc}_6$ & 6 & 10 & 1/21 & 22 & 5.76 & 973.17\\
		$\mathit{AMBAArb}_2$ & - & $\mathit{AMBAArb}_4$ & - & 17 & - & - & timeout & timeout\\
		$\mathit{AMBAArb}_4$ & - & $\mathit{AMBAArb}_6$ & - & 23 & - & - & timeout & timeout\\
		$\mathit{AMBALock}_2$ & - & $\mathit{AMBALock}_4$ & - & 5 & - & - & timeout & timeout\\
		$\mathit{AMBALock}_4$ & - & $\mathit{AMBALock}_6$ & - & 5 & - & - & timeout & timeout\\
		\hline
		$\mathit{LoadFull}_2$ & 3 & $\mathit{LoadFull}_3$ & 6 & 21 & 1/10 & 10 & 6.50 & 49.67\\
		$\mathit{Loadfull}_3$ & 6 & $\mathit{LoadFull}_4$ & - & 25 & -  & - & timeout & timeout\\
		\hline
		$\mathit{BitStuffer}_2$ & 5 & $\mathit{BitStuffer}_3$ & 7 & 7 & 2/7 & 9 & 0.08 & 1.02\\
		$\mathit{BitStuffer}_3$ & 7 & $\mathit{BitStuffer}_4$ & 9 & 9 & 1/6 & 7 & 0.21 & 3.97\\
		\hline
		$\mathit{ABPRec}_1$ & 2 & $\mathit{ABPRec}_2$ & 4 & 9 & 2/5 & 7 & 0.11 & 1.52\\
		$\mathit{ABPRec}_2$ & 4 & $\mathit{ABPRec}_3$ & 8 & 17 & 4/9 & 13 & 2.87 & 326.98\\
		$\mathit{ABPTran}_1$ & 2 & $\mathit{ABPTran}_2$ & 4 & 31 & 1/5 & 5 & 0.76 & 76.93\\
		$\mathit{ABPTran}_2$ & 4 & $\mathit{ABPTran}_3$ & - & 91 & - & - & timeout & timeout\\
		\hline
		$\mathit{TCP}_1$ & 2 & $\mathit{TCP}_2$ & 4 & 6 & 1/4 & 5 & 0.05 & 0.19\\
		$\mathit{TCP}_2$ & 4 & $\mathit{TCP}_3$ & 8 & 8 & 3/14 & 17 & 0.58 & 14.47\\
		\hline
		$\mathit{Scaling}_4$ & 4 & $\mathit{Scaling}'_4$ & 4 & 4 & 4/0 & 4 & 0.02 & 0.10\\
		$\mathit{Scaling}_5$ & 5 & $\mathit{Scaling}'_5$ & 5 & 5 & 4/0 & 4 & 0.03 & 0.22\\
		$\mathit{Scaling}_6$ & 6 & $\mathit{Scaling}'_6$ & 6 & 6 & 4/0 & 4 & 0.04 & 0.51\\
		$\mathit{Scaling}_8$ & 8 & $\mathit{Scaling}'_8$ & 8 & 8 & 4/0 & 4 & 0.12 & 2.54\\
		$\mathit{Scaling}_{12}$ & 12 & $\mathit{Scaling}'_{12}$ & 12 & 12 & 4/0 & 4 & 34.02 & 167.03\\
		\hline
	\end{tabular}
	}
\end{center}
	\vskip -.9cm
\end{table}

\subsection{Technical Details}
We instantiate BoSy with an explicit encoding, a linear search strategy, an input-symbolic backend and moore semantics to match our implementation. Both tools only synthesize winning strategies for system players. We use \textbf{ltl3ba}\cite{ltl3ba} as the converter from an LTL-specification to an automaton for both tools. As both constraint systems only contain existential quantifiers, \textbf{CryptoMiniSat}\cite{cryptoMiniSat} is used as the SAT-solver. The solution bound is the minimal bound that is given as input and the initial transition system is synthesized using BoSy, at first. The benchmarks results were obtained on a single quad-core Intel Xeon processor (E3-1271 v3, 3.6GHz) with 32 GB RAM, running a GNU/Linux system. A timeout of 2 hours is used.

\subsection{Observations}

The benchmark results are shown in \Cref{table:Benchmarks}. For each benchmark, the table contains two specifications, an initial and an extended one. For example, the initial specification for the first benchmark $\mathit{Arb}_2$ specifies a two-client arbiter and the extended one $\mathit{Arb}_4$ specifies a four-client arbiter. Additionally, the table records the minimal solution bound for each of the specifications. Our explainable synthesis protoype starts by synthesizing a system for the initial specification and then synthesizes a minimal repair and an explanation for the extended one. If the minimal repair has to access additional states, that are initially unreachable, our prototype initializes them with a self loop for all input assignments where no output holds. The traditional synthesis tool BoSy only synthesizes a solution for the extended specification. The size of the universal co-B\"uchi automaton, representing the extended specification is reported. For the explainable synthesis, the applied operations of the minimal repair and the number of justifications is given. For both tools, the runtime is reported in seconds.

The benchmark results reveal that our explainable synthesis approach is able to solve the same benchmarks like BoSy. In all cases, except two, we are able to synthesize explanations where every operation can be single justified. The applied operations show that there are only a small number of changed state labelings, primarily for reaching additional states. Since only minimal initial systems are synthesized, the outputs in the given structure are already fixed. Redirecting transitions repairs the system  more efficient. In general, the evaluation reveals that the runtime for the explainable synthesis process takes a multiple of BoSy with respect to the number of applied operations. Thus, the constraint-based synthesis for minimal repairs is not optimal if a small number of operations is sufficient since the repair synthesis problem is polynomial in the size of the system as proven in \Cref{theorem:decision_minimal_transformation_fixed_state_space}. To improve the runtime and to solve more instances, many optimizations, used in existing synthesis tools, can be implemented. These extensions include different encodings such as QBF or DQBF or synthesizing strategies for the environment or a mealy semantics.

\section{Conclusion}\label{sec:Conclusion}

In this paper, we have developed an explainable synthesis process for reactive systems. For a set of specification, expressed in LTL, the algorithm incrementally builds an implementation by repairing intermediate implementations to satisfy the currently considered specification. In each step, an explanation is derived to justify the taken changes to repair an implementation.   We have shown that the decision problem of finding a repair for a fixed number of transformations is polynomial in the size of the system and exponential in the number of operations. By extending the constraint system of bounded synthesis, we can synthesize minimal repairs where the resulting system is size-optimal. We have presented an algorithm that constructs explanations by using Bounded Model Checking to obtain counterexample traces. The evaluation of our prototype showed that explainable synthesis, while more expensive, can still solve the same benchmarks as a standard synthesis tool. In future work, we plan to develop this approach into a comprehensive tool that provides rich visual feedback to the user. Additionally, we plan to investigate further types of explanations,  including quantitive and symbolic explanations.

\bibliographystyle{plain}
\bibliography{mybib}

\begin{thebibliography}{10}

\bibitem{ltl3ba}
Tom{\'a}{\v{s}} Babiak, Mojm{\'i}r K{\v{r}}et{\'i}nsk{\'y}, Vojt{\v{e}}ch
  {\v{R}}eh{\'a}k, and Jan Strej{\v{c}}ek.
\newblock Ltl to b{\"u}chi automata translation: Fast and more deterministic.
\newblock In Cormac Flanagan and Barbara K{\"o}nig, editors, {\em Tools and
  Algorithms for the Construction and Analysis of Systems}, pages 95--109,
  Berlin, Heidelberg, 2012. Springer Berlin Heidelberg.

\bibitem{BMC}
Armin Biere, Alessandro Cimatti, Edmund~M Clarke, Ofer Strichman, Yunshan Zhu,
  et~al.
\newblock Bounded model checking.
\newblock {\em Advances in computers}, 58(11):117--148, 2003.

\bibitem{programRepair}
Borzoo Bonakdarpour and Bernd Finkbeiner.
\newblock Program repair for hyperproperties.
\newblock In Yu-Fang Chen, Chih-Hong Cheng, and Javier Esparza, editors, {\em
  Automated Technology for Verification and Analysis}, pages 423--441, Cham,
  2019. Springer International Publishing.

\bibitem{encodingsOfBoundedSynthesis}
Peter Faymonville, Bernd Finkbeiner, Markus Rabe, and Leander Tentrup.
\newblock Encodings of bounded synthesis.
\newblock pages 354--370, 03 2017.

\bibitem{BoSy}
Peter Faymonville, Bernd Finkbeiner, and Leander Tentrup.
\newblock Bosy: An experimentation framework for bounded synthesis.
\newblock In {\em Proceedings of {CAV}}, volume 10427 of {\em LNCS}, pages
  325--332. Springer, 2017.

\bibitem{lazySynthesis}
Bernd Finkbeiner and Swen Jacobs.
\newblock Lazy synthesis.
\newblock In Viktor Kuncak and Andrey Rybalchenko, editors, {\em Verification,
  Model Checking, and Abstract Interpretation - 13th International Conference,
  {VMCAI} 2012, Philadelphia, PA, USA, January 22-24, 2012. Proceedings},
  volume 7148 of {\em Lecture Notes in Computer Science}, pages 219--234.
  Springer, 2012.

\bibitem{BoundedCycle}
Bernd Finkbeiner and Felix Klein.
\newblock Bounded cycle synthesis.
\newblock volume 9779 of {\em Lecture Notes in Computer Science}. Springer
  Berlin Heidelberg, 2016.

\bibitem{outputsensitive}
Bernd Finkbeiner and Felix Klein.
\newblock Reactive synthesis: Towards output-sensitive algorithms.
\newblock In Alexander Pretschner, Doron Peled, and Thomas Hutzelmann, editors,
  {\em Dependable Software Systems Engineering}, volume~50 of {\em {NATO}
  Science for Peace and Security Series, {D:} Information and Communication
  Security}, pages 25--43. {IOS} Press, 2017.

\bibitem{BoundedSynthesis}
Bernd Finkbeiner and Sven Schewe.
\newblock Bounded synthesis.
\newblock {\em International Journal on Software Tools for Technology
  Transfer}, 15(5-6):519--539, 2013.

\bibitem{skeletons}
Bernd Finkbeiner and Hazem Torfah.
\newblock Synthesizing skeletons for reactive systems.
\newblock In Cyrille Artho, Axel Legay, and Doron Peled, editors, {\em
  Automated Technology for Verification and Analysis - 14th International
  Symposium, {ATVA} 2016 Proceedings}, Lecture Notes in Computer Science, 2016.

\bibitem{SYNTCOMP2017}
Swen Jacobs, Nicolas Basset, Roderick Bloem, Romain Brenguier, Maximilien
  Colange, Peter Faymonville, Bernd Finkbeiner, Ayrat Khalimov, Felix Klein,
  Thibaud Michaud, Guillermo~A. Perez, Jean-Francois Raskin, Ocan Sankur, and
  Leander Tentrup.
\newblock The 4th reactive synthesis competition {(SYNTCOMP 2017)}: Benchmarks,
  participants and results.
\newblock In {\em SYNT 2017}, volume 260 of {\em EPTCS}, pages 116--143, 2017.

\bibitem{DBLP:journals/corr/abs-1904-07736}
Swen Jacobs, Roderick Bloem, Maximilien Colange, Peter Faymonville, Bernd
  Finkbeiner, Ayrat Khalimov, Felix Klein, Michael Luttenberger, Philipp~J.
  Meyer, Thibaud Michaud, Mouhammad Sakr, Salomon Sickert, Leander Tentrup, and
  Adam Walker.
\newblock The 5th reactive synthesis competition {(SYNTCOMP} 2018): Benchmarks,
  participants {\&} results.
\newblock {\em CoRR}, abs/1904.07736, 2019.

\bibitem{programrepairgame}
Barbara Jobstmann, Andreas Griesmayer, and Roderick Bloem.
\newblock Program repair as a game.
\newblock In Kousha Etessami and Sriram~K. Rajamani, editors, {\em Computer
  Aided Verification, 17th International Conference, {CAV} Proceedings}.
  Springer, 2005.

\bibitem{crest17}
Hadas Kress{-}Gazit and Hazem Torfah.
\newblock The challenges in specifying and explaining synthesized
  implementations of reactive systems.
\newblock In {\em Proceedings CREST@ETAPS}, {EPTCS}, pages 50--64, 2018.

\bibitem{safralessSynthesis}
Orna Kupferman, Nir Piterman, and Moshe~Y. Vardi.
\newblock Safraless compositional synthesis.
\newblock In Thomas Ball and Robert~B. Jones, editors, {\em Computer Aided
  Verification, 18th International Conference, {CAV} 2006, Seattle, WA, USA,
  August 17-20, 2006, Proceedings}, volume 4144 of {\em Lecture Notes in
  Computer Science}, pages 31--44. Springer, 2006.

\bibitem{refinement1}
P.~Nilsson and N.~Ozay.
\newblock Incremental synthesis of switching protocols via abstraction
  refinement.
\newblock In {\em 53rd IEEE Conference on Decision and Control}, pages
  6246--6253, 2014.

\bibitem{refinement2}
Hans-J{\"o}rg Peter and Robert Mattm{\"u}ller.
\newblock Component-based abstraction refinement for timed controller
  synthesis.
\newblock In Theodore Baker, editor, {\em Proceedings of the 30th IEEE
  Real-Time Systems Symposium (RTSS 2009), December 1 - December 4, 2009,
  Washington, D.C., USA}, pages 364--374, Los Alamitos, CA, USA, December 2009.
  IEEE Computer Society.

\bibitem{article}
A.~Pnueli and Roni Rosner.
\newblock On the synthesis of a reactive module.
\newblock {\em Automata Languages and Programming}, 372:179--190, 01 1989.

\bibitem{LTL}
Amir Pnueli.
\newblock The temporal logic of programs.
\newblock In {\em Proceedings of the 18th Annual Symposium on Foundations of
  Computer Science}, SFCS ’77, page 46–57, USA, 1977. IEEE Computer
  Society.

\bibitem{refinement3}
G.~{Reissig}, A.~{Weber}, and M.~{Rungger}.
\newblock Feedback refinement relations for the synthesis of symbolic
  controllers.
\newblock {\em IEEE Transactions on Automatic Control}, 62(4):1781--1796, 2017.

\bibitem{termite}
Leonid Ryzhyk and Adam Walker.
\newblock Developing a practical reactive synthesis tool: Experience and
  lessons learned.
\newblock In Ruzica Piskac and Rayna Dimitrova, editors, {\em Proceedings Fifth
  Workshop on Synthesis, SYNT@CAV 2016, Toronto, Canada, July 17-18, 2016},
  volume 229 of {\em {EPTCS}}, pages 84--99, 2016.

\bibitem{clarke}
A.~Sistla and Edmund Clarke.
\newblock The complexity of propositional linear temporal logics.
\newblock {\em J. ACM}, 32:733--749, 07 1985.

\bibitem{cryptoMiniSat}
Mate Soos, Karsten Nohl, and Claude Castelluccia.
\newblock Extending sat solvers to cryptographic problems.
\newblock In Oliver Kullmann, editor, {\em Theory and Applications of
  Satisfiability Testing - SAT 2009}, pages 244--257. Springer Berlin
  Heidelberg, 2009.

\end{thebibliography}

\end{document}